\documentclass[12pt,a4paper]{article}
\usepackage[T1]{fontenc}
\usepackage{lmodern}
\usepackage{amsmath,amssymb,amsthm,mathtools,bm}
\usepackage{microtype}
\usepackage{geometry}
\usepackage{hyperref}
\usepackage{cite}
\usepackage{booktabs}
\usepackage{xcolor}
\usepackage{graphicx}
\usepackage{float}
\hypersetup{colorlinks=true,linkcolor=blue!55!black,citecolor=blue!55!black,urlcolor=blue!55!black}
\newcommand{\R}{\mathbb R}
\newcommand{\C}{\mathbb C}
\newcommand{\Z}{\mathbb Z}

\newcommand{\dd}{\mathrm d}
\newcommand{\e}{\mathrm e}
\newcommand{\Tr}{\mathrm{Tr}}

\newcommand{\Hom}{\mathrm{Hom}}
\newcommand{\End}{\mathrm{End}}

\newcommand{\cH}{\mathcal H}
\newcommand{\cQ}{\mathcal Q}
\newcommand{\cD}{\mathcal D}
\newcommand{\cM}{\mathcal M}

\newcommand{\cF}{\mathcal F}

\newcommand{\phys}{\mathrm{phys}}

\newcommand{\ADHM}{\mathrm{ADHM}}

\newcommand{\diag}{\mathrm{diag}}
\newtheorem{theorem}{Theorem}
\newtheorem{proposition}{Proposition}

\newtheorem{remark}{Remark}
\newtheorem{definition}{Definition}
\title{\large\textbf{The Chiral Random-Matrix Ensemble of the Type-IIB Axion--Dilaton Wormhole Partition Function}}
\author{Soo-Jong Rey\\[0.2cm]
{\sl Kwangwoon University}\\
{\sl Seoul, Korea}}
\date{\tt sjrey@kw.ac.kr}
\begin{document}
\maketitle
\begin{abstract}
I construct a microscopic ADHM/chiral-Wishart representative of the reduced charge-sector coefficient \(W_\nu[b]\) that enters the Type-IIB axion--dilaton wormhole partition function \(Z_{\rm wh}(\theta;b)\). I fix the axion-charge sector, equivalently the form-field-flux sector, which plays the exact same structural role as a fixed-topology sector in QCD: it supplies the domain where I compute the microscopic coefficient before forming the final theta sum. At \(E=0\), the axion--dilaton radial family reaches its BPS instanton endpoint. After I impose the Hamiltonian constraint, gauge quotient, charge-sector boundary condition, and collective zero-mode quotient, the physical quadratic fluctuation operator at that endpoint becomes a positive adjoint square. Its non-zero spectrum is therefore a squared singular-value spectrum, and its microscopic endpoint is a Laguerre/chiral hard edge. The D\((-1)\)/D3 super-ADHM collective-coordinate integral supplies the Type-IIB microscopic representative of this same coefficient. In the large-\(N\) result of Dorey et al., this super-ADHM measure becomes the \(k\)-D-instanton measure on \(AdS_5\times S^5\), multiplied by a centered zero-dimensional supersymmetric matrix-model factor. The chiral/Wishart ensemble forms the hard-edge limit of the rectangular block inside this super-ADHM integral. Fermionic ADHM variables and supergravity fermions remain part of the coefficient: in protected sectors they cancel paired non-zero modes, impose zero-mode saturation, and determine which reduction data \(b\) give a non-vanishing \(W_\nu[b]\).
\end{abstract}
\tableofcontents

\section{Introduction}
\label{sec:intro}
In companion works~\cite{Rey2026a, Rey2026b}, built upon previous works~\cite{GiddingsStrominger1988,Coleman1988,Rey1989a,Rey1989b,Rey1991,Rey1999} I formulated the Type-IIB axion--dilaton wormhole in sectors of fixed conserved charge. This charge is the axion charge, or equivalently the form-field flux. I performed the Hessian calculation at one such sector: its admissible variations remain tangent to the fixed-charge domain, while a charge-changing variation moves to a different saddle. This sector ordering is the starting point of the present paper.

At fixed charge, the solution with first integral \(E=0\) is the BPS instanton. The solutions with \(E>0\) are non-BPS wormholes with smooth throats connecting non-compact parent universes. These solutions belong to one radial family, but they define different fluctuation problems. The BPS instanton is the endpoint of the family. A complete non-BPS wormhole, a neck-cut geometry, and a long-distance two-end operator term each carry their own boundary conditions, boundary terms, contours, and self-adjoint domains. I define the operator below as the physical quadratic fluctuation operator of the BPS endpoint in a fixed charge sector.

Let \(\Phi_0(\nu)\) denote the BPS instanton at \(E=0\) in charge sector \(\nu\). Let \(\cF\) be the Witten--Nester/Prasad--Sommerfield--Bogomolny (BPS) map~\cite{Witten1981,Nester1981,Prasad1975,Bogomolnyi1976} derived from the axion Routhian. After I solve the Hamiltonian constraint, quotient gauge directions, impose the charge-sector boundary condition, and separate collective zero modes, the quadratic action on the physical endpoint domain takes the singular-value-square form~\cite{Rey2026a}
\begin{equation}
\delta^2 S_{\phys}=\langle \eta,\cH_\nu\eta\rangle,
\qquad
\cH_\nu=\cQ_\nu^\dagger\cQ_\nu,
\qquad
\cQ_\nu\eta=
\left.{\dd\over\dd\epsilon}\,
\cF\!\left[\Phi_0(\nu)+\epsilon\eta\right]\right|_{\epsilon=0},
\quad \eta\in\mathcal D_{\nu,\phys}.
\label{eq:intro-square}
\end{equation}
Here \(\cQ_\nu\) is the Fr\'echet derivative of the BPS map at \(\Phi_0(\nu)\), restricted to the physical endpoint domain \(\mathcal D_{\nu,\phys}\). The Hamiltonian constraint~\cite{Sasaki1986} removes the conformal-factor direction before Eq.~\eqref{eq:intro-square} defines the operator. I formalize this as a theorem regarding the physical endpoint Hessian of the BPS instanton in a fixed charge sector. The unreduced Euclidean Hessian of an \(E>0\) non-BPS wormhole is a separate operator.

One solves the eigenvalue equation for \(\cQ_\nu^\dagger\cQ_\nu\) and gets the squared singular values of the first-order map \(\cQ_\nu\). The spectrum terminates at zero, and eigenvalues of this physical endpoint operator remain on the non-negative axis. Zero is therefore a hard spectral edge. Any finite-mode approximation that preserves the Fredholm index, the anti-unitary symmetry class, and the endpoint boundary form lies in a Laguerre, or chiral, random-matrix universality class. The Bessel kernel appears at the microscopic edge because the local variable is a singular value of a rectangular chiral block.

What I refer to as the chirality is operator-theoretic. In QCD, chiral random-matrix theory models the Euclidean Dirac operator and its fermion determinant~\cite{VerbaarschotZahed1993,Verbaarschot1994,VerbaarschotWettig2000,LeutwylerSmilga1992,NagaoNishigaki2000}. Here the fluctuating fields are bosonic supergravity fields, and the Gaussian fluctuation factor contains the inverse square root of a determinant. The shared structure is the first-order singular-value map; the determinant power belongs to the coefficient weight. The canonical double
\begin{equation}
\cD_\nu=\begin{pmatrix}0&\cQ_\nu^\dagger\\[2pt]\cQ_\nu&0\end{pmatrix},
\qquad
\cD_\nu^2=\begin{pmatrix}\cQ_\nu^\dagger\cQ_\nu&0\\[2pt]0&\cQ_\nu\cQ_\nu^\dagger\end{pmatrix}
\label{eq:intro-double}
\end{equation}
has the same block structure as a chiral Dirac operator, while the Type-IIB path-integral interpretation remains bosonic. The random-matrix statement is a statement about singular values of a first-order map. This  chiral matrix ensemble belongs to the general classification of~\cite{AltlandZirnbauer1997,Zirnbauer1996}.

The ADHM system~\cite{ADHM1978,AHS1978} for D\((-1)\)/D3 brane complex supplies the microscopic Type-IIB realization of that structure~\cite{Dorey1999,Dorey2000,Dorey2002}. The ADHM data contain a rectangular bifundamental block. Singular-value gauge fixing of this block gives a Laguerre measure, with Dyson index fixed by the real, complex, or quaternionic structure of the block~\cite{James1954,Hua1963,Muirhead1982,Forrester2010,Andreief1883,BumpDiaconis2002,Szego1952}. This finite-dimensional construction realizes the same chiral singular-value normal form as the continuum BPS endpoint Hessian.

To make this identification, I need a precise label dictionary. The ADHM instanton number is \(k\). The rectangular defect of the chosen ADHM block is \(r\). The axion charge sector is \(\nu\). The hard-edge index of the gravitational BPS endpoint Hessian is \(a_\nu\). These labels carry different definitions. A compactification or duality map can identify them only after I fix the charge dictionary, endpoint domain, and anti-unitary symmetry class.

The matrix ensemble represents a coefficient in the Type-IIB axion--dilaton wormhole partition function. In the companion coefficient analysis~\cite{Rey2026a, Rey2026b}, I write the wormhole partition function in the theta representation as
\begin{equation}
Z_{\rm wh}(\theta;b)=\sum_{\nu\in\Z} W_\nu[b]e^{i\nu\theta}.
\label{eq:intro-B-coefficient}
\end{equation}
Figure~\ref{fig:C-parent-universe-two-end} displays the coefficient-level sequence. I reduce the two-end matrix to the scalar coefficient \(W_\nu[b]\), and then assemble the reduced coefficients into the theta-representation wormhole partition function.

\begin{figure}[H]
\centering
\includegraphics[width=0.92\textwidth]{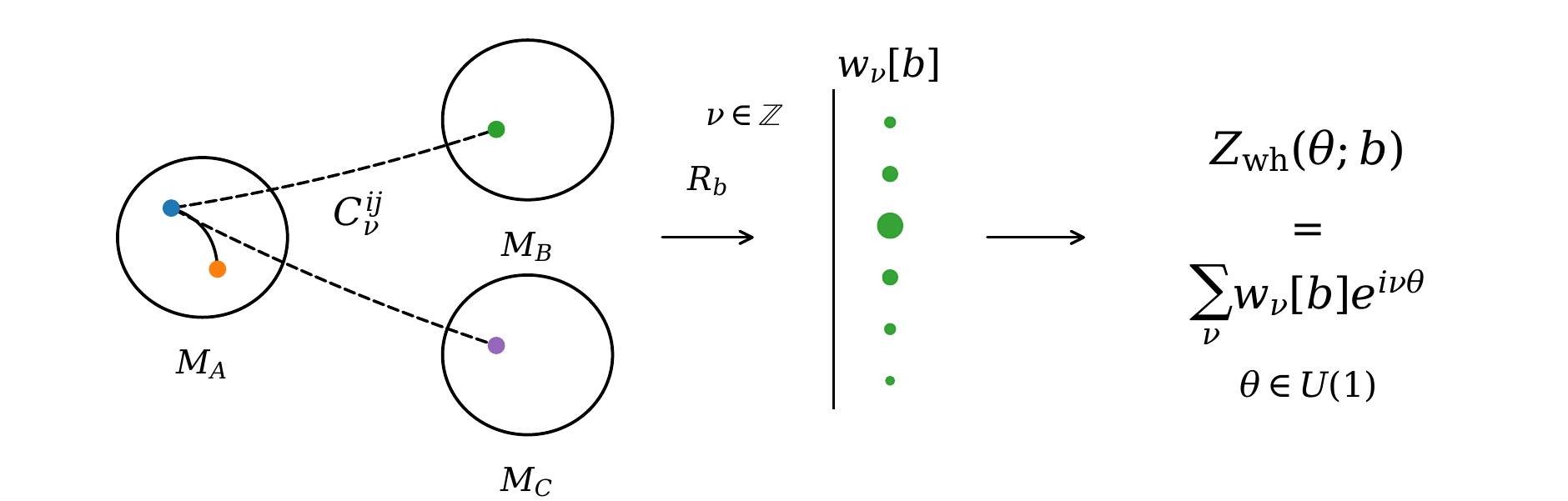}
\caption{Two-end operator reduction in the notation of the coefficient analysis. The parent-universe labels $A,B$ specify where the two integrated end insertions are placed, while $i,j$ label the end-insertion operators. The connecting coefficient is the matrix $C_\nu^{ij}$ in a chosen charge sector. The reduction data $b$ specify the operation $R_b$ that turns the unreduced coefficient matrix into the scalar reduced coefficient $W_\nu[b]$. I then sum the reduced coefficients with the compact character $e^{i\nu\theta}$ to give the wormhole partition function in the theta representation $Z_{\rm wh}(\theta;b)$.}
\label{fig:C-parent-universe-two-end}
\end{figure}

The symbol \(b\) denotes the reduction data. It records the asymptotic data, trace or final insertion, projection, zero-mode insertion, source normalization, contour, and reduction operation used to reduce the two-end coefficient. The matrix ensemble I study here is a microscopic representative of \(W_\nu[b]\) in a fixed charge sector. I obtain the full function \(Z_{\rm wh}(\theta;b)\) only after summing these coefficients over charge sectors.

I make the comparison with QCD exact at the level of sector logic. Chiral random-matrix theory represents the fixed-topology low-energy sector of the QCD path integral below a matching scale; the ultraviolet gauge theory supplies the sectors, normalization, and matching data. Here the input is a fixed Type-IIB charge sector and its BPS endpoint square. The output is the ADHM/chiral-Wishart representative of the reduced coefficient that enters the charge-sector sum for the wormhole partition function.

I organize the rest of this paper as follows. In Section~\ref{sec:bps-to-hard-group}, I derive the hard spectral edge from the BPS endpoint square. In Section~\ref{sec:adhm-group}, I derive the finite-dimensional chiral block from the ADHM construction. In Section~\ref{sec:grav-matching-group}, I compare the Type-IIB endpoint with the ADHM block through the charge map, hard-edge index, and antiunitary class. In Section~\ref{sec:verification-group}, I state the consistency conditions, verification tests, consequences, and extensions of the matching problem.

\section{From the BPS endpoint to the chiral hard spectral edge}
\label{sec:bps-to-hard-group}
\subsection{Chiral structure of the BPS endpoint}
\label{sec:crm-background}
Here, ``chiral random-matrix structure'' simply means a singular-value structure. A chiral random matrix is the universal finite-dimensional model for a first-order map and for the squared singular values of that map. If a map \(M\) connects two spaces of unequal dimensions, the doubled operator
\begin{equation}
D_M=\begin{pmatrix}0&M^\dagger\\ M&0\end{pmatrix}
\end{equation}
has an exact chiral grading, while \(D_M^2\) contains \(M^\dagger M\) and \(MM^\dagger\). The non-zero spectrum of the adjoint square is a singular-value spectrum. Its endpoint at zero is hard because the eigenvalues cannot pass through zero into the negative axis.

This links QCD, the D\((-1)\)/D3 system, and the Type-IIB BPS endpoint. In QCD, \(M\) is the off-diagonal block of a Euclidean Dirac operator, and the topological charge fixes the number of exact zero modes. In the ADHM problem, \(M\) is a bifundamental instanton--brane block, and its singular values are microscopic instanton--brane coordinates. In the Type-IIB BPS endpoint problem, \(M\) is replaced by the first variation of the BPS map on the physical charge-sector domain. The field content and determinant powers differ across these examples, but the singular-value geometry is the same.

The singular-value geometry fixes the universality class. A Hermitian random matrix with an unconstrained endpoint has a soft spectral edge and Airy statistics. A positive square has a spectral endpoint at zero and Laguerre/Bessel statistics. The charge-sector endpoint analysis places the BPS endpoint Hessian of the wormhole coefficient in the second class. Figure~\ref{fig:C-airy-hard-edge} displays this distinction. Below, I derive the universal spectral law implied by the endpoint square and then test the ADHM block as its microscopic Type-IIB realization.

\begin{figure}[H]
\centering
\includegraphics[width=0.88\textwidth]{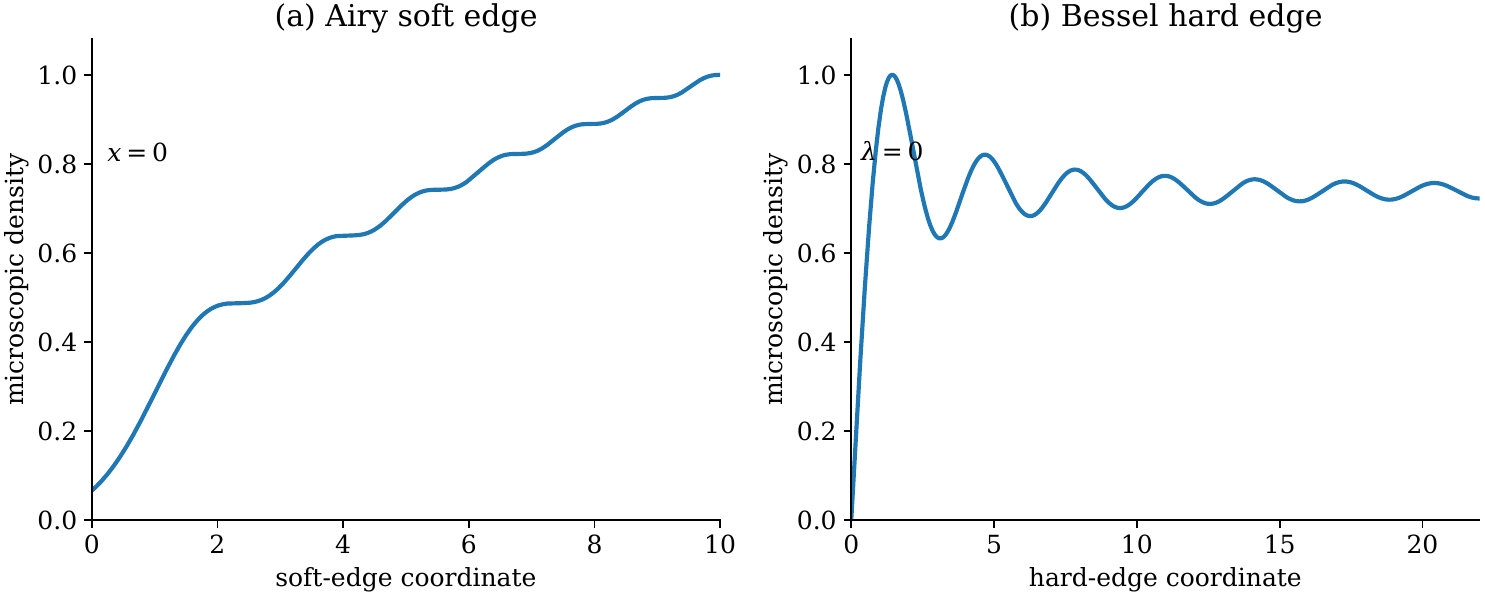}
\caption{Soft edge versus hard spectral edge. Panel (a) shows the Airy soft-edge density with the local coordinate chosen so that the spectral side lies at positive values of the local coordinate. Panel (b) shows the Bessel hard-edge density at a non-negative spectral endpoint. The BPS limit Hessian belongs to the second class because it is an adjoint square on the physical charge-sector domain. The relevant microscopic variables are squared singular values at a hard endpoint.}
\label{fig:C-airy-hard-edge}
\end{figure}

I must keep two Bessel mechanisms separate. Axion charge sectors give modified Bessel functions through Fourier analysis of \(\theta\) on a circle. A chiral hard spectral edge gives Bessel kernels through singular-value scaling at zero. These mechanisms have different origins, and their labels become related only after I match the charge, index, and microscopic domain.

The hard spectral edge is the spectral consequence of the BPS instanton endpoint. A non-BPS wormhole moves away from this endpoint by changing the radial first integral. At fixed charge \(\nu\), the BPS endpoint operator has spectrum bounded below by zero, so its eigenvalues accumulate at a fixed spectral endpoint. The corresponding universal ensemble is Laguerre/chiral, with Bessel edge statistics in place of Airy soft-edge statistics.

A singular value has a direct physical meaning in this endpoint problem. It measures the quadratic cost of a physical perturbation at the BPS instanton. The first-order BPS map measures the failure of the perturbation to satisfy the BPS equation, and the Hessian is the adjoint square of that failure. A small singular value identifies a near-BPS deformation direction toward the non-BPS wormhole branch. An exact zero singular value is a collective or gauge direction and is treated outside the determinant. The random-matrix variable is therefore the spectral coordinate of a near-BPS deformation.

\subsection{Charge sector and BPS limit Hessian}
\label{sec:paperA-input}
I state the spectral property of the BPS endpoint in the terminology of the companion analysis~\cite{Rey2026a}.

\begin{definition}[BPS limit Hessian]
Fix the conserved axion charge sector, equivalently the form-field flux sector, and let \(\Phi_0(\nu)\) be the \(E=0\) BPS instanton in that sector. Let \(\mathcal D_{\nu,\phys}\) denote the physical fluctuation domain obtained by imposing the Hamiltonian constraint, fixing gauge redundancies, imposing charge-sector boundary conditions, and removing collective zero modes. The BPS limit Hessian is
\begin{equation}
\cH_\nu=\cQ_\nu^\dagger\cQ_\nu,
\qquad
\cQ_\nu\eta=
\left.{\dd\over\dd\epsilon}\,
\cF\!\left[\Phi_0(\nu)+\epsilon\eta\right]\right|_{\epsilon=0},
\quad \eta\in\mathcal D_{\nu,\phys}.
\label{eq:BPS-limit-Hessian}
\end{equation}
Thus \(\cQ_\nu\) is the Fr\'echet derivative of the BPS map at the BPS instanton, restricted to the physical endpoint domain.
\end{definition}

This definition fixes four pieces of data. The charge sector is fixed. The operator is evaluated at the \(E=0\) BPS instanton; a generic \(E>0\) non-BPS wormhole has its own fluctuation operator. The Hessian \(\cH_\nu\) is the Schur-complement, or radial Sasaki--Mukhanov, physical Hessian after I remove the conformal-factor direction. The adjoint is defined with the inner product and boundary condition of the charge-sector endpoint problem.

\begin{proposition}[Endpoint positivity]
On \(\mathcal D_{\nu,\phys}\), the spectrum of the BPS limit Hessian is non-negative:
\begin{equation}
\operatorname{Spec}(\cH_\nu)\subset [0,\infty).
\label{eq:nonnegative-spectrum}
\end{equation}
The point \(0\) is the hard spectral edge.
\end{proposition}
\begin{proof}
For any physical perturbation \(\eta\) in the domain,
\begin{equation}
\langle \eta,\cH_\nu\eta\rangle
=\langle \cQ_\nu\eta,\cQ_\nu\eta\rangle\ge 0.
\end{equation}
The zero modes are collective BPS directions and are removed by quotienting or treated separately. The inequality holds for the resulting physical endpoint operator and is separate from the full fluctuation problem of an \(E>0\) non-BPS saddle.
\end{proof}

This positivity statement identifies the spectral problem of this paper. The hard spectral edge follows from the operator equation~\eqref{eq:BPS-limit-Hessian}; the throat picture and dilute instanton-gas approximation play no role in deriving it. The non-BPS wormholes with \(E>0\) enter as deformations away from the endpoint, and their full Hessians have their own boundary forms, ensembles, and contours.

The physical interpretation follows immediately. The BPS instanton is the endpoint of the radial family. A non-BPS wormhole begins when an admissible physical deformation opens away from that endpoint. The square-root operator measures the first-order failure of the BPS equations along such a deformation, and the Hessian measures the square of that failure. The spectrum near zero is the spectrum of near-BPS deformation directions. The hard-edge density counts how densely these directions accumulate after I impose charge conservation, gauge quotienting, the Hamiltonian constraint, and the collective-coordinate quotient.

\subsection{From the adjoint square to a chiral hard spectral edge}
\label{sec:square-to-hard}
Let $\cQ_\nu$ be a closed Fredholm representative of the BPS map on the endpoint domain. One solves the eigenvalue equation $\cQ_\nu^\dagger\cQ_\nu u_j=s_j^2u_j$ and gets the non-zero singular values $s_j>0$. The eigenvalues of $\cH_\nu$ are $\lambda_j=s_j^2$. Thus $\lambda=0$ is the hard spectral edge of the adjoint square.

For a finite-mode truncation, I replace $\cQ_\nu$ by a matrix $W$ between two finite-dimensional physical spaces. If the two spaces have dimensions $n$ and $n+a$, then $a$ singular values are forced to zero on one side of the chiral double. The positive eigenvalues $\lambda_i=s_i^2$ have the Laguerre joint density
\begin{equation}
P_{a,\beta}(\lambda)
=\frac{1}{Z}\prod_{i<j}|\lambda_i-\lambda_j|^\beta
\prod_i \lambda_i^{\frac\beta2(a+1)-1}
\e^{-\frac\beta4 V(\lambda_i)}.
\label{eq:Laguerre-density}
\end{equation}
Here $\beta=1,2,4$ records the real, complex, or quaternionic antiunitary class. The potential $V$ fixes the non-universal regulator and scale. The microscopic kernel at the origin depends only on the hard-edge data. For $\beta=2$, after the standard hard-edge scaling, the kernel is the Bessel kernel
\begin{equation}
K_a(\zeta,\xi)=
\frac{\sqrt{\zeta\xi}}{\zeta^2-\xi^2}
\left[
\zeta J_{a+1}(\zeta)J_a(\xi)-
\xi J_{a+1}(\xi)J_a(\zeta)
\right].
\label{eq:Bessel-kernel}
\end{equation}
The diagonal density is
\begin{equation}
\rho_{s,a}(\zeta)
=\frac{\zeta}{2}\left[J_a(\zeta)^2-J_{a+1}(\zeta)J_{a-1}(\zeta)\right]
\sim \frac{\zeta^{2a+1}}{\Gamma(a+1)\Gamma(a+2)}.
\label{eq:Bessel-density}
\end{equation}
The exponent $2a+1$ is the most direct observable of the hard-edge index. Figure~\ref{fig:C-hard-edge-density} shows the microscopic density for several indices.

\begin{figure}[H]
\centering
\includegraphics[width=0.82\textwidth]{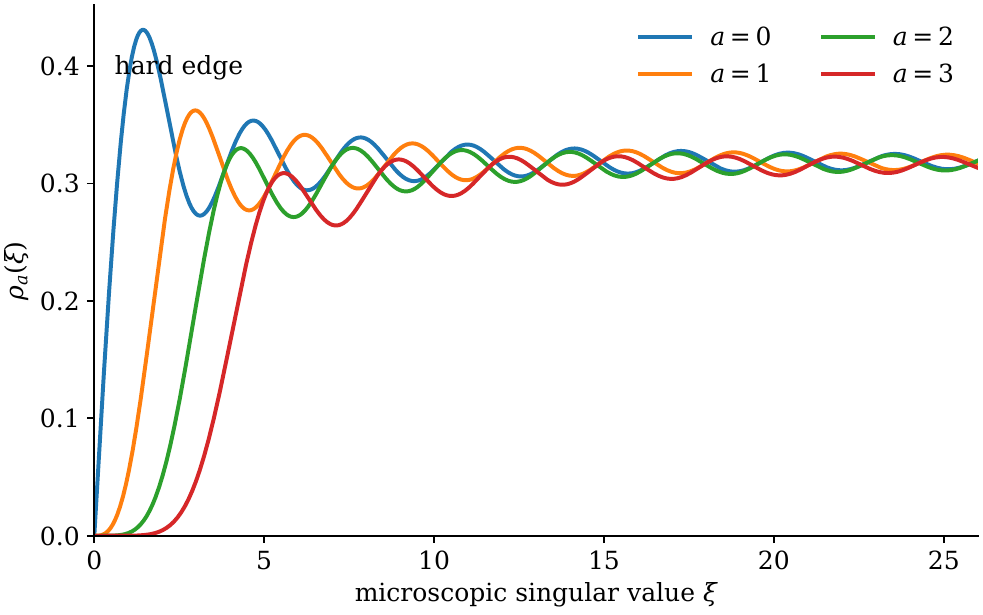}
\caption{Microscopic hard-edge density of the chiral/Wishart ensemble for several values of the hard-edge index. The origin is a spectral endpoint of a non-negative singular-value problem. Increasing the index strengthens the repulsion from zero. In the Type-IIB application this index is the hard-edge index of the BPS limit map after the charge sector, endpoint domain, zero-mode quotient, antiunitary class, and reduction data have been matched.}
\label{fig:C-hard-edge-density}
\end{figure}

\begin{theorem}[BPS square to Laguerre hard spectral edge]
Assume that a sequence of finite-dimensional truncations of the BPS limit map $\cQ_\nu$ preserves its Fredholm index, antiunitary class, and self-adjoint endpoint boundary form. Then the microscopic spectrum of the physical Hessian $\cH_\nu=\cQ_\nu^\dagger\cQ_\nu$ at zero lies in the Laguerre, or chiral, hard-edge universality class with index
\begin{equation}
a_\nu=\operatorname{ind}_{\rm edge}\cQ_\nu.
\label{eq:a-nu}
\end{equation}
For $\beta=2$ the limiting kernel is Eq.~\eqref{eq:Bessel-kernel}; for $\beta=1,4$ it is the corresponding Pfaffian Bessel kernel.
\end{theorem}
\begin{proof}
Equation~\eqref{eq:BPS-limit-Hessian} identifies the physical spectrum with the squared singular values of a first-order map. A finite-mode truncation that preserves the index gives a rectangular chiral block. The change of variables from matrix entries to squared singular values gives the Laguerre density~\eqref{eq:Laguerre-density}. Hard-edge universality for Laguerre ensembles then gives the Bessel kernel. The gravitational input is the endpoint analysis, which identifies the physical map $\cQ_\nu$ and its domain.
\end{proof}

This theorem separates the BPS endpoint from a soft Airy edge. An Airy edge is the soft endpoint of an unconstrained Hermitian spectrum. At the BPS endpoint, the adjoint-square structure fixes zero as a structural endpoint of the spectrum.

I assign a wormhole interpretation to each factor in the Laguerre density. The Vandermonde factor gives repulsion among independent near-BPS endpoint deformation directions: coincident singular directions reduce the orbit volume of the physical fluctuation space. The power of each eigenvalue at the origin records the unpaired directions of the square-root map. The exponential factor is a finite-mode or ultraviolet regulator. The endpoint observables are the hard-edge exponent, the Bessel kernel class, and the microscopic scale fixed by the compactification.

The ensemble provides the local normal form of one fixed charge-sector endpoint problem. The finite matrix replaces the infinite-dimensional map while preserving the index, antiunitary class, and boundary form. Universality enters because the microscopic edge discards the detailed regulator and retains the endpoint at zero together with the number of unpaired square-root directions.

\subsection{Bosonic determinant and chiral singular values}
\label{sec:bosonic}
We must distinguish two uses of the term ``chiral''. In QCD it refers to the anticommutation of the Dirac operator with $\gamma_5$ and to the fermion determinant. In the BPS endpoint Hessian it refers to the block form of the doubled square-root operator~\eqref{eq:intro-double}. In the present setting, chirality means the singular-value structure of the square root of a bosonic endpoint Hessian.

The Gaussian integrals have different determinant powers. A bosonic physical Hessian contributes schematically
\begin{equation}
Z_{\nu,\rm fluc}^{\rm bos}
\sim (\det{}'\cH_\nu)^{-1/2}
=(\det{}'\cQ_\nu^\dagger\cQ_\nu)^{-1/2},
\label{eq:bos-det}
\end{equation}
with zero modes and collective coordinates treated separately. A fermionic chiral Dirac operator contributes
\begin{equation}
Z_{\rm ferm}\sim \det(\cD+m).
\label{eq:ferm-det}
\end{equation}
The determinant powers and zero-mode insertions differ, while the microscopic singular-value kernel has the same mathematical form because both problems reduce near zero to squared singular values.

The BPS endpoint Hessian is bosonic, and its square root is a first-order map. That first-order map has a chiral double. The hard-edge classification follows from this double.

The bosonic determinant changes the one-loop weight, especially the weight assigned to very small eigenvalues. It leaves the geometric origin of the eigenvalues unchanged. The same singular values that enter a chiral Dirac determinant in QCD enter here through an inverse square-root determinant after I separate zero modes. The QCD correspondence is therefore a correspondence of singular-value universality. Determinant powers, condensates, and fermionic zero-mode saturation are fixed by the Type-IIB degrees of freedom in the coefficient calculation.

Equation~\eqref{eq:bos-det} supplies one factor in the final coefficient. In a supersymmetric BPS sector, paired non-zero bosonic and fermionic fluctuations cancel, up to ghosts, endpoint domains, anomalies, and lifted modes. The inverse square root displays the bosonic square. The complete coefficient $W_\nu[b]$ also contains the supergravity fermion Pfaffian, ghost determinants, collective-coordinate Jacobians, and fermionic zero-mode insertions. In protected sectors these factors reduce the non-zero-mode ratio and leave a finite-dimensional supermoduli integral. The microscopic model is therefore the super-ADHM integral, with the purely bosonic Wishart integral appearing as its hard-edge bosonic sector.

\subsection{The sector coefficient represented by the matrix model}
\label{sec:wnu-target}
The matrix model targets the reduced coefficient defined in the coefficient analysis. Let $C_\nu^{ij}$ be the unreduced coefficient matrix of the far-distance two-end source term in charge sector $\nu$. The labels $i,j$ refer to end-insertion operators, while parent-universe labels record the placement of those insertions. The reduction data $b$ reduce this coefficient matrix in end-insertion label space to the scalar coefficient
\begin{equation}
W_\nu[b]=R_b\!\left(C_\nu^{ij}\right).
\label{eq:w-reduction-C}
\end{equation}
I form the theta series after this reduction. The matrix integral in this paper represents the fixed-$\nu$ calculation that produces $W_\nu[b]$, or equivalently the unreduced coefficient before $R_b$ is applied. The matrix integral is part of the coefficient calculation before the charge sum forms $Z_{\rm wh}(\theta;b)$.

The gravitational expression for $W_\nu[b]$ is a Type-IIB path integral in a chosen axion, or form-field, charge sector. Schematically,
\begin{equation}
W_\nu[b]
=R_b\left[
\int_{\nu}Dg\,D\Phi\,D\Psi\,D({\rm ghosts})
\exp(-S_{\rm IIB,Routhian})
\right].
\label{eq:wnu-IIB-path}
\end{equation}
Here $\Phi$ denotes bosonic supergravity fields and $\Psi$ denotes the fermionic fields. The formula is schematic because the asymptotic data, contour, zero-mode insertions, collective-coordinate measure, source normalization, and reduction operation are part of $b$. This expression is the counterpart of the fixed-topology QCD gauge-field integral. In QCD, chiral random-matrix theory replaces the universal microscopic part of the fixed-topology gauge-field integral. Here the ADHM/chiral-Wishart matrix model replaces the universal microscopic part of the fixed-charge Type-IIB integral that contributes to $W_\nu[b]$.

This identification also fixes the role of the classical action. The chiral or Wishart matrix integral represents the microscopic sector contribution, with the classical action matched as part of the coefficient. After matching, the matrix model represents the microscopic sector coefficient, including the action normalization, the supermoduli measure, and the zero-mode saturation appropriate to the chosen observable. The overall normalization and the precise map from the Type-IIB charge sector to the finite-dimensional matrix-model parameters are matching data. The rectangular hard-edge structure is the universal part.

\section{ADHM microscopic description and finite-dimensional chiral ensembles}
\label{sec:adhm-group}
\subsection{Instanton calculus and the finite-dimensional chiral block}
\label{sec:instanton-calculus}
The ADHM construction provides a finite-dimensional model in which the singular-value variables appear before any continuum spectral limit. For $k$ D$(-1)$ branes and $N$ D3 branes, the bifundamental fields are matrices connecting the instanton vector space and the colour vector space. Their rectangularity produces the Laguerre exponent at the origin.

The supersymmetric ADHM quotient contains adjoint variables, bifundamentals, constraints, gauge orbits, and fermionic partners. The present analysis uses the localized radial part of the bifundamental block. Localization in the supersymmetric D$(-1)$/D3 system reduces this radial part to a singular-value integral, and that integral is the finite Laguerre ensemble. The real, complex, or quaternionic structure of the bifundamental fixes the Dyson index. The rank defect fixes the rectangular exponent. The finite-volume closed forms are Toeplitz or Pfaffian Bessel functions.

The gravitational endpoint operator provides the continuum square, while the ADHM variables provide a microscopic Type-IIB realization of a chiral block. The comparison is a matching problem between two singular-value descriptions. The charge label, instanton number, rectangular defect, and hard-edge index retain distinct definitions until the matching theorem relates them.

I can test this dictionary. The ADHM side predicts a Dyson class, a rectangular exponent, finite-size Laguerre polynomials, and Bessel/Toeplitz closed forms. The gravitational side must reproduce the same hard-edge index and antiunitary class from the BPS endpoint map. A mismatch places the two systems in different microscopic hard-edge classes, even though both remain Laguerre-type singular-value problems.

I can give the ADHM block a direct wormhole interpretation. D$(-1)$/D3 strings provide finite-dimensional coordinates for charged instanton degrees of freedom before those degrees of freedom appear as continuum radial fluctuations. Their singular values are microscopic radial variables in the charged instanton moduli problem, distinct from macroscopic throat radii. These variables can survive at the BPS endpoint as hard-edge coordinates. ADHM therefore gives a microscopic Type-IIB origin for the same singular-value geometry that the BPS endpoint Hessian exposes in the continuum. The ADHM contribution remains tied to $W_\nu[b]$ because it describes the microscopic sector integral before I assemble the charge-sector coefficients into $Z_{\rm wh}(\theta;b)$.

\subsection{The large-$N$ super-ADHM bridge to $AdS_5\times S^5$}
\label{sec:dorey-largeN}
The large-$N$ result of Dorey, Hollowood, Khoze, Mattis, and Vandoren supplies the microscopic Type-IIB anchor. In $\mathcal N=4$ $SU(N)$ Yang--Mills theory, the complete $k$-instanton collective-coordinate integral is the supersymmetric ADHM integral. In the large-$N$ limit at fixed instanton number $k$, the collective-coordinate space collapses to a single copy of $AdS_5\times S^5$. The $AdS_5$ coordinates are the instanton position and size, and the $S^5$ coordinates arise from the auxiliary variables used to treat the fermionic quadrilinear interaction. The centered degrees of freedom contribute the partition function $Z_k$ of ten-dimensional $\mathcal N=1$ $SU(k)$ super-Yang--Mills theory reduced to zero dimensions \cite{Dorey1999,Dorey2000,Dorey2002}.

This result identifies the holographic quantity. The super-ADHM integral computes the $k$-D-instanton contribution to protected $AdS_5\times S^5$ correlators. For correlators that saturate the sixteen exact supersymmetric and superconformal fermion zero modes, the large-$N$ result has the form of a Type-IIB D-instanton Witten diagram: a D-instanton effective vertex is integrated over $AdS_5\times S^5$ and attached to boundary operators by bulk-to-boundary propagators. Schematically,
\begin{equation}
G_n\big|_{k{\rm -inst}}
\sim e^{2\pi i k\tau} Z_k
\int_{AdS_5\times S^5}
d^4x_0\,\frac{d\rho}{\rho^5}\,d\Omega_5\,d^{16}\Xi\,
\mathcal I_n(x_\ell;x_0,\rho,\Omega,\Xi).
\label{eq:Dorey-AdS-measure}
\end{equation}
The exact power of $N$, the power of $k$, and the divisor sum depend on the protected correlator. The structural conclusion is independent of these details: in this protected sector, the large-$N$ ADHM integral is the semiclassical Type-IIB D-instanton path integral on $AdS_5\times S^5$.

I use this result as the microscopic representative of the protected Type-IIB charge-sector coefficient. The coefficient entering the wormhole partition function is a super-ADHM collective-coordinate integral after the reduction data $b$ have selected the insertions, saturated the fermion zero modes, fixed the contour, and reduced the two-end coefficient. In the Dorey setting this gives the template
\begin{equation}
w_k[b]
\sim e^{2\pi i k\tau} Z_k
\int_{AdS_5\times S^5}d\mu_{\rm D-inst}\,d^{16}\Xi\,
\mathcal I_b.
\label{eq:wnu-Dorey-template}
\end{equation}
In this sense, the ADHM matrix model represents the charge-sector coefficient. The chiral/Wishart hard spectral edge is a further reduction of the rectangular sector of this supermatrix integral. The full supergravity coefficient also contains the collective-coordinate integral, zero-mode saturation, insertions, and reduction data. Equation~\eqref{eq:wnu-Dorey-template} gives the microscopic Type-IIB representative of the sector coefficient that the coefficient analysis later sums over charge.

This result also fixes the role of fermions. In a protected BPS sector, paired non-zero bosonic and fermionic modes cancel. The surviving object is the supermoduli integral with exact fermion zero modes, constrained ADHM variables, and the centered matrix-model partition function $Z_k$. Reduction data that fail to saturate the exact zero modes give a vanishing coefficient. Reduction data that saturate them give a D-instanton contribution to a protected Type-IIB amplitude. This zero-mode-saturated supermoduli integral is the microscopic object compared with the reduced coefficient $W_\nu[b]$.

\subsection{ADHM data and the Dyson index}
\label{sec:adhm-data}
I start the microscopic Type-IIB realization with $k$ D-instantons, or D$(-1)$-branes, in the background of $N$ D3-branes. In the unitary case the ADHM data are
\begin{equation}
B_1,B_2\in \End(\C^k),
\qquad
I\in\Hom(\C^k,\C^N),
\qquad
J\in\Hom(\C^N,\C^k),
\end{equation}
subject to
\begin{align}
\mu_\C&=[B_1,B_2]+IJ=0,\label{eq:ADHM-complex}\\
\mu_\R&=[B_1,B_1^\dagger]+[B_2,B_2^\dagger]+II^\dagger-J^\dagger J
=\zeta\,\mathbf 1_k.\label{eq:ADHM-real}
\end{align}
The moduli space is the hyper-Kahler quotient
\begin{equation}
\cM_{k,N}=\mu^{-1}(0,\zeta)/G_k.
\end{equation}
For a $U(N)$ gauge group its real dimension is $4kN$. The integer $k$ is the instanton number. The instanton number and the rectangular defect of a matrix block are distinct data.

The bifundamental block carries the reality type. In the three standard cases,
\begin{center}
\begin{tabular}{c c c c c}
\toprule
D3 gauge group & reality of $I$ & instanton group $G_k$ & Dyson $\beta$ & chiral class \\
\midrule
$U(N)$ & complex & $U(k)$ & $2$ & AIII \\
$SO(N)$ & real & $O(k)$ & $1$ & BDI \\
$Sp(N)$ & quaternionic & $Sp(k)$ & $4$ & CII \\
\bottomrule
\end{tabular}
\end{center}
The table records the antiunitary class of the bifundamental singular-value problem. I match the labels $k$, $N$, axion charge, and hard-edge index separately.

\subsection{Singular-value gauge fixing of the bifundamental}
\label{sec:svd}
The instanton gauge group acts on the right of $I$. Singular-value gauge fixing removes this redundancy directly. For $N\ge k$ write
\begin{equation}
I=U\Sigma V^\dagger,
\qquad
\Sigma=\begin{pmatrix}\diag(\sigma_1,\ldots,\sigma_k)\\0\end{pmatrix},
\qquad
\sigma_a\ge 0.
\label{eq:SVD}
\end{equation}
Here $V\in G_k$ is gauged, $U$ is a colour rotation, and the gauge-invariant radial variables are
\begin{equation}
\lambda_a=\sigma_a^2.
\end{equation}
The rectangular defect is
\begin{equation}
r=N-k \qquad (N\ge k),
\label{eq:r-defect}
\end{equation}
with the transpose convention used when $k>N$.

For the complex case, the flat metric is $\dd s^2=\Tr(\dd I^\dagger\dd I)$. Differentiating Eq.~\eqref{eq:SVD} and moving to the body-fixed frame gives
\begin{equation}
\delta I=\omega\Sigma+\dd\Sigma+\Sigma\eta,
\qquad
\omega=U^\dagger\dd U\in\mathfrak u(N),
\qquad
\eta=\dd V^\dagger V\in\mathfrak u(k).
\end{equation}
For an unordered pair $a\neq b$ in the upper $k\times k$ block,
\begin{equation}
\begin{pmatrix}(\delta I)_{ab}\\[2pt]\overline{(\delta I)_{ba}}\end{pmatrix}
=
\begin{pmatrix}\sigma_b&\sigma_a\\ \sigma_a&\sigma_b\end{pmatrix}
\begin{pmatrix}\omega_{ab}\\ \eta_{ab}\end{pmatrix}.
\end{equation}
As a real map this contributes $|\sigma_b^2-\sigma_a^2|^2$. The lower block contributes $\sigma_a^{2r}$. Converting $\sigma_a$ to $\lambda_a$ gives
\begin{equation}
\dd I=c\,
\prod_{a<b}(\lambda_a-\lambda_b)^2
\prod_a\lambda_a^{r}
\prod_a\dd\sigma_a\,\dd U\,\dd V
=c'\,
\prod_{a<b}(\lambda_a-\lambda_b)^2
\prod_a\lambda_a^{r}
\prod_a\lambda_a^{-1/2}\dd\lambda_a\,\dd U\,\dd V.
\end{equation}
I perform the same orbit-volume computation in the real and quaternionic cases and get the James--Hua formula
\begin{equation}
\boxed{
\dd I=c_{N,k,\beta}
\prod_{a<b}|\lambda_a-\lambda_b|^\beta
\prod_a\lambda_a^{\frac\beta2(r+1)-1}
\prod_a\dd\lambda_a\,\dd U\,\dd V.}
\label{eq:James-Hua}
\end{equation}
Equation~\eqref{eq:James-Hua} gives the singular-value Jacobian that underlies the Laguerre measure. The Vandermonde factor is kinematic: it is the Riemannian volume of the gauge orbit written in singular-value coordinates. The power of $\lambda$ is fixed by the rectangular defect. The derivation uses only the singular-value geometry of the bifundamental block.

For the wormhole problem, this formula is the finite-dimensional prototype of the endpoint measure. The Vandermonde factor means that independent charged near-BPS deformation directions repel because the gauge orbit degenerates when two singular values coincide. The factor $\lambda^{\frac\beta2(r+1)-1}$ records how many directions on one side of the first-order block lack partners on the other side. In the continuum BPS problem this role is played by the edge index of $\cQ_\nu$. The ADHM calculation therefore states exactly what the gravitational computation must supply: the endpoint, the index, and the real structure of the physical endpoint domain.

\begin{remark}
The singular-value slice is adapted to the microscopic throat-opening variables. It exposes the chiral random-matrix structure of the ADHM bifundamental directly. This construction is the main technical input from the source notes retained here.
\end{remark}

\subsection{Localization and the finite Laguerre integral}
\label{sec:localization}
In the maximally supersymmetric D$(-1)$/D3 system, the ADHM matrix model has a nilpotent supercharge, denoted here by $\mathsf Q$ to distinguish it from the BPS map $\cQ_\nu$. The exact equality below belongs to that supersymmetric setting. In a less supersymmetric axionic background, the equality requires the same cancellations to survive.

The adjoint variables $B_1,B_2$ and their superpartners cancel off the kernel. The complex moment-map constraint is imposed by a $\mathsf Q$-exact pair, and the real moment-map equation is solved by the singular-value variables. The remaining radial integral is the bifundamental sector. With masses $x_f$ for the relevant flavour insertions one obtains
\begin{equation}
Z^{\ADHM}_{r,\beta}(\{x_f\})
=C\int_0^\infty\left(\prod_{a=1}^{s}\dd\lambda_a\right)
\prod_{a<b}|\lambda_a-\lambda_b|^\beta
\prod_a\lambda_a^{\frac\beta2(r+1)-1}
\e^{-\frac\beta4\lambda_a}
\prod_f(\lambda_a+x_f^2).
\label{eq:ADHM-Laguerre}
\end{equation}
Here
\begin{equation}
s=\min(N,k),
\qquad
r=|N-k|.
\end{equation}
Equation~\eqref{eq:ADHM-Laguerre} is a finite chiral random-matrix integral. Its topological-looking exponent is the rectangular defect $r$, while the physical instanton number remains $k$.

I must keep this distinction clear for the gravitational application. The axion charge sector is $\nu$. The gravitational hard-edge index of the BPS endpoint Hessian is $a_\nu$. The ADHM rectangularity is $r$. A holographic dictionary can set $a_\nu=r$ in a specified sector only after I derive that equality from the charge map and endpoint domain.

Equation~\eqref{eq:ADHM-Laguerre} gives the precise sense in which ADHM produces chiral RMT. The localized matrix model provides a regulated ensemble of charged microscopic modes at the instanton endpoint. Its hard spectral edge is the degeneracy locus of the bifundamental singular values, and its Bessel index is the rectangular defect. If the gravitational BPS endpoint Hessian flows to the same edge, the wormhole endpoint inherits this microscopic brane structure. If the gravitational edge carries different data, the ADHM integral remains a hard-edge model with a different index or symmetry class.

\subsection{Supersymmetric matter, Wishart integrals, and Bessel functions}
\label{sec:sqcd-wishart}
The relevant supersymmetric gauge-theory precedent comes from supersymmetric QCD with fundamental matter. Pure super-Yang--Mills contains a gluino and leads, as a Dirac-spectrum problem, to the usual chiral random-matrix classes with the fermion in the adjoint representation. Supersymmetric QCD with fundamental matter contains quarks, squarks, and gauginos, and the matrix-model description of the matter sector naturally involves constrained rectangular variables. In the Dijkgraaf--Vafa treatment of SQCD matter, the fundamental variables $Q$ and $\tilde Q$ are replaced by gauge-invariant mesonic variables, and the constrained matrix integral is equivalent to a complex Wishart integral. This Wishart measure gives the matter part of the Affleck--Dine--Seiberg and Veneziano--Yankielowicz--Taylor structure \cite{DemasureJanik2003}.

This precedent matters because a Wishart integral is the Laguerre hard-edge ensemble used in chiral random-matrix theory. When a rectangular matrix $X$ is replaced by $X^\dagger X$, the eigenvalue density contains a Vandermonde factor and a power of each eigenvalue at the origin. The hard-edge scaling limit is governed by Bessel functions. Different observables display different Bessel functions: ordinary Bessel functions appear in spectral kernels, modified Bessel $I$ functions appear in compact fixed-sector partition functions, and modified Bessel $K$ functions can appear in noncompact or inverse-determinant sectors. A constrained Wishart integral therefore has a Bessel microscopic hard-edge regime even when a chosen holomorphic observable hides the Bessel function.

I apply this lesson to the ADHM problem in a specific way. The rectangular D$(-1)$/D3 bifundamental becomes a Wishart-type variable after singular-value gauge fixing. In the full supersymmetric problem it is embedded in a super-ADHM measure with fermionic constraints and zero modes. The bosonic rectangular block provides the Laguerre/Bessel hard-edge limit. The fermionic variables determine the numerator, the zero-mode saturation rule, and the possible selection-rule zeros. That SQCD lesson is the one needed for the Type-IIB ADHM coefficient.

\subsection{Worked ADHM examples and finite-size examples}
\label{sec:examples}
The rectangular data are elementary, but I record them to prevent later index errors. The following table uses the unitary case and the convention $r=|N-k|$, $s=\min(N,k)$.
\begin{center}
\begin{tabular}{c c c c p{4.3cm}}
\toprule
$(N,k)$ & $\dim \cM_{k,N}$ & $r=|N-k|$ & $s$ & singular-value content \\
\midrule
$(2,1)$ & $8$ & $1$ & $1$ & one radial eigenvalue, no Vandermonde \\
$(3,1)$ & $12$ & $2$ & $1$ & one radial eigenvalue with larger rectangular exponent \\
$(2,2)$ & $16$ & $0$ & $2$ & two eigenvalues with Vandermonde repulsion \\
$(4,2)$ & $32$ & $2$ & $2$ & two eigenvalues and non-zero hard-edge index \\
$(3,3)$ & $36$ & $0$ & $3$ & square block with three singular values \\
$(5,2)$ & $40$ & $3$ & $2$ & two singular values and rectangular defect three \\
\bottomrule
\end{tabular}
\end{center}
The dimension $4kN$ is the ADHM moduli-space dimension in the unitary case. The number of singular values is $s$. The Laguerre exponent is controlled by $r$. The instanton number is $k$. These quantities enter different parts of the finite-size example.

Finite size introduces a second distinction, which I must track. At fixed finite $k$, the integral~\eqref{eq:ADHM-Laguerre} is a Gaussian-damped polynomial in the masses. The leading small-mass term of that finite polynomial can differ from the microscopic Bessel power. The Bessel order is a statement about the chiral microscopic limit. The finite ADHM integral and the limiting chiral partition function are therefore connected by the standard double scaling, while the first nonzero power of a finite polynomial by itself carries less information.

These examples have direct wormhole use because they separate three integer labels. The number of microscopic radial variables is controlled by $s$. The suppression or enhancement near the singular-value endpoint is controlled by $r$. The gravitational charge sector is controlled by $\nu$. A throat calculation that uses $k$ as the Bessel order, or identifies $N-k$ with axion charge before deriving the charge map, collapses distinct data that the endpoint problem requires. The finite examples function as calibration cases for the coefficient calculation: they show which integer labels enter the ADHM block before I assert any matching to the gravitational charge sector.

\subsection{The Toeplitz determinant from the Andr\'eief identity}
\label{sec:andreief-derivation}
The determinant formula below gives the finite-volume form of the chiral unitary partition function. It fixes the Bessel minors that appear in the microscopic representative of $W_\nu[b]$. For the unitary class, the microscopic flavour integral is
\begin{equation}
Z_r^{(N_f)}(x)=\int_{U(N_f)}\dd\mu(U)\,(\det U)^r
\exp\left[\frac{x}{2}\Tr(U+U^{-1})\right].
\label{eq:unitary-start-andreief}
\end{equation}
Diagonalizing $U=V\operatorname{diag}(e^{i\theta_1},\ldots,e^{i\theta_{N_f}})V^\dagger$ gives
\begin{equation}
Z_r^{(N_f)}(x)= \frac{1}{N_f!}
\int \prod_{a=1}^{N_f}\frac{\dd\theta_a}{2\pi}
\prod_{a<b}|e^{i\theta_a}-e^{i\theta_b}|^2
\prod_a e^{ir\theta_a+x\cos\theta_a}.
\label{eq:unitary-eigen-angle}
\end{equation}
Writing the Vandermonde as two determinants and applying the Andr\'eief identity gives
\begin{equation}
Z_r^{(N_f)}(x)=
\det_{1\le i,j\le N_f}\left[
\int_0^{2\pi}\frac{\dd\theta}{2\pi}
e^{i(r+i-j)\theta+x\cos\theta}
\right]
=\det_{1\le i,j\le N_f}\left[I_{r+i-j}(x)\right].
\label{eq:andreief-toeplitz}
\end{equation}
The determinant is Toeplitz because the matrix element depends on $i-j$. A determinant with $i+j$ is a Hankel determinant and computes a different integral. This algebraic distinction separates the charge-sector Bessel function from the chiral finite-volume partition function.

The Toeplitz structure records compact angular variables in the chiral finite-volume problem. It is the chiral finite-volume remnant of the unitary zero-mode integral. I obtain the hard-edge density after analytic continuation to the spectral axis and microscopic scaling. The same Bessel functions appear in both places, but the observables are different.

In the wormhole language, this is where compactness re-enters the microscopic edge. The angle integral is the zero-mode angular integral of the chiral finite-volume problem, distinct from the spacetime theta dependence of \(Z_{\rm wh}(\theta;b)\). The Toeplitz determinant says that the hard-edge limit retains compact phase memory, just as the axionic sector sum retains compact charge memory. The first structure counts the spectral response of endpoint modes inside a chosen charge sector. The second structure counts charge transfer between the two ends. Both structures enter the coefficient problem at different steps. This separation keeps the special-function formulas attached to their physical objects: Toeplitz determinants belong to the microscopic edge of $W_\nu[b]$, while the theta sum belongs to the final wormhole partition function.

\subsection{Colour--flavour transformation and Bessel closed forms}
\label{sec:toeplitz}
The finite Laguerre integral admits the standard colour--flavour representation. For $\beta=2$, the microscopic limit with rescaled masses gives the unitary coset integral
\begin{equation}
Z_r^{\beta=2,N_f}(x)=
\int_{U(N_f)}\dd\mu(U)\,(\det U)^r
\exp\left[\frac{x}{2}\Tr(U+U^{-1})\right].
\label{eq:unitary-coset}
\end{equation}
At one flavour this reduces to
\begin{equation}
Z_r(x)=\int_0^{2\pi}\frac{\dd\alpha}{2\pi}
\e^{ir\alpha+x\cos\alpha}=I_r(x).
\label{eq:one-flavour-Bessel}
\end{equation}
The function $I_r$ in Eq.~\eqref{eq:one-flavour-Bessel} is a modified Bessel function in a chiral finite-volume partition function. A charge-sector Skellam law also contains modified Bessel functions. The two appearances are related by compact Fourier analysis, but equality of their labels is a matching condition. The ADHM label is $r$; the axion charge label is $\nu$.

The colour--flavour transformation converts the radial singular-value problem into an angular zero-mode problem. Physically, this transformation changes coordinates on the microscopic fluctuation integral while leaving the wormhole saddle fixed. On one side the endpoint modes appear as radial singular values; on the other side they are encoded in compact flavour angles. This mirrors the larger axion/form-field relation while keeping the two transformations distinct: compact variables and charge variables become interchangeable only after I specify the transform, contour, asymptotic data, and reduction data.

For general $N_f$ and $\beta=2$, I diagonalize $U$ and apply the Andr\'eief identity to get the Toeplitz determinant
\begin{equation}
\boxed{
Z_r^{\beta=2,N_f}(x)=
\det_{1\le i,j\le N_f}\left[I_{r+i-j}(x)\right].}
\label{eq:Toeplitz}
\end{equation}
The $i-j$ dependence is the character-theoretic Toeplitz structure of the unitary integral. The Hankel determinant $\det[I_{r+i+j}]$ computes a different object.

At two flavours, I write explicit closed forms for the three chiral symmetry classes. With the normalization used in the source notes,
\begin{align}
Z_r^{\beta=2}(x)&=I_r(x)^2-I_{r+1}(x)I_{r-1}(x),
\label{eq:beta2-two}\\
Z_r^{\beta=1}(x)&=\frac{1}{3}\left[I_{r-2}(x)I_{r+2}(x)-4I_{r-1}(x)I_{r+1}(x)+3I_r(x)^2\right],
\label{eq:beta1-two}\\
Z_r^{\beta=4}(x)&=\int_0^1 I_{2r}(2xt)\,\dd t
=\sum_{m\ge0}\frac{I_{r+m}(x)I_{r-m}(x)}{2m+1}.
\label{eq:beta4-two}
\end{align}
The $\beta=2$ expression is the Bessel Turan determinant. The orthogonal and symplectic formulas are Pfaffian/coset reductions written in scalar form at two flavours. They display the dependence on the reality class.

\begin{proposition}[Ordering in the two-flavour normalization]
For the positive normalization of Eqs.~\eqref{eq:beta1-two}--\eqref{eq:beta4-two}, numerical evaluation of the defining coset integrals and of the closed forms gives
\begin{equation}
0<Z_r^{\beta=1}(x)<Z_r^{\beta=2}(x)<Z_r^{\beta=4}(x)
\end{equation}
for the tested range of non-negative integer $r$ and positive $x$. The ordering records the increasing eigenvalue repulsion as the Vandermonde power rises.
\end{proposition}

The ordering provides a diagnostic of the three Dyson classes and catches sign or normalization mistakes in the coset formulas.

\subsection{Orthogonal and symplectic reductions}
\label{sec:orth-symp}
The unitary class gives the cleanest display of the Toeplitz determinant, but Type-IIB orientifold variants naturally produce the other two chiral classes. In the orthogonal case, the bifundamental can be taken real and the quotient leaves a chiral BDI ensemble. In the symplectic case, the bifundamental is quaternionic and the quotient leaves a chiral CII ensemble. The singular values remain non-negative in all three cases. The orbit volume changes, and with it the Vandermonde power.

I write the finite-dimensional eigenvalue density as
\begin{equation}
\prod_{a<b}|\lambda_a-\lambda_b|^\beta
\prod_a\lambda_a^{\frac\beta2(r+1)-1}w(\lambda_a)\,\dd\lambda_a,
\qquad \beta=1,2,4.
\label{eq:three-beta-density}
\end{equation}
For $\beta=1$ and $\beta=4$, the corresponding microscopic correlations are Pfaffian in place of determinantal. The kernels are Bessel kernels in the appropriate Pfaffian sense. The same Dyson index that changes bulk eigenvalue repulsion also changes the edge exponent.

The antiunitary class forms part of the hard-edge data. If the gravitational BPS endpoint map has a real structure, the ADHM comparison must use the orthogonal chiral class. If it has a quaternionic structure, the comparison must use the symplectic class. A matching theorem must therefore match the antiunitary class as well as the Bessel order.

\begin{proposition}[Dyson-class matching criterion]
A proposed ADHM realization of the gravitational BPS hard spectral edge must match both the integer hard-edge index and the antiunitary class. Equality of rectangularity and gravitational index is insufficient if the real, complex, or quaternionic structures differ.
\end{proposition}

The Dyson index records how charge conjugation, orientation, and possible orientifold reality conditions act on the microscopic square root. It is part of the definition of the endpoint universality class.

The same statement holds in gravity. The antiunitary structure is part of the self-adjoint endpoint problem. Two BPS endpoint maps with the same index can have different hard-edge statistics if their real structures are real, complex, or quaternionic. A Type-IIB realization describes the same microscopic edge only when the real structure of the ADHM block matches the real structure of the continuum Hessian.

\section{Gravitational interpretation and matching}
\label{sec:grav-matching-group}
\subsection{The BPS limit Hessian as the gravitational hard spectral edge}
\label{sec:grav-hard-edge}
The gravitational operator is the physical BPS endpoint Hessian in a specified charge sector. This $E=0$ operator is the adjoint square~\eqref{eq:BPS-limit-Hessian}. The hard spectral edge is the origin of the spectrum of this physical endpoint operator. The full non-BPS wormhole Hessian belongs to the finite-throat problem.

Let $a_\nu$ be the hard-edge index of $\cQ_\nu$. In a Fredholm model, $a_\nu$ is the difference between the dimensions of the two endpoint kernels after the zero-mode quotient. In a finite-mode approximation, $a_\nu$ is the rectangular defect of the matrix representing the map. In the radial Sasaki--Mukhanov problem, the charge-sector boundary condition and the self-adjoint domain determine $a_\nu$. The endpoint operator therefore computes the index.

The chiral double~\eqref{eq:intro-double} gives eigenvalues $\pm s_j$ away from zero. Squaring gives the physical eigenvalues $s_j^2$ of $\cH_\nu$. The microscopic density at zero is then Eq.~\eqref{eq:Bessel-density}, with $a=a_\nu$ and with the Dyson index fixed by the antiunitary class of the endpoint problem. This construction is the gravitational chiral hard spectral edge.

The distinction from bosonic JT gravity is clear. The standard JT matrix integral~\cite{SaadShenkerStanford2019,StanfordWitten2019} has a soft spectral edge described by Airy universality. The Type-IIB BPS endpoint Hessian has a hard endpoint because the operator is an adjoint square. A soft spectral edge can move under deformations of the potential. The hard spectral edge at zero is protected by the square and by the chosen charge-sector domain.

The random-matrix class therefore admits a geometric reading. In JT-like soft-edge problems, an equilibrium density determines an endpoint that can move when the potential changes. In the BPS endpoint problem, the non-negativity of a physical square pins the endpoint at zero. The relevant distinction is between an equilibrium edge and a kinematic spectral endpoint. The axion--dilaton BPS endpoint belongs to the kinematic case.

\subsection{Radial endpoint problem and neck intuition}
\label{sec:radial-neck}
The radial problem displays the endpoint square most directly. The charge-sector Routhian converts the axion into a fixed flux variable and produces a first-order BPS equation at $E=0$. Linearizing the BPS equation gives the first-order map $\cQ_\nu$. The Hamiltonian constraint eliminates the lapse and non-dynamical metric fluctuations and leaves the Sasaki--Mukhanov, or Schur-complement, physical variable. The quadratic action of that variable is the adjoint square.

An $E>0$ non-BPS wormhole is a finite throat with a neck. Cutting the throat creates an artificial boundary. The scalar momentum at that boundary can fail to vanish, so the cut half-geometry is a constrained instanton unless the correct neck data and Legendre term are imposed. These operations belong to the full non-BPS spectrum and to the construction of the long-distance two-end operator term. The BPS endpoint Hessian is defined before those operations enter.

The radial coordinate also clarifies the hard spectral edge. The non-extremality parameter measures the departure from the BPS instanton endpoint. On the physical endpoint domain, the quadratic fluctuation energy is a norm of a first-order variation. Hence the local spectral coordinate is a squared singular value. The hard spectral edge is the operator-theoretic image of the BPS inequality.

Geometrically, the BPS instanton is the tip of the radial family. The non-BPS wormhole branch opens away from that tip, while the physical endpoint Hessian measures the squared slope of admissible directions there. Singular values are the principal slopes of this deformation map. A small singular value means that the BPS equations are violated weakly along that direction, so the near-BPS throat deformation has a small quadratic cost. The hard-edge exponent measures how many such near-BPS deformations remain after I remove the exact BPS moduli and gauge directions.

The throat picture provides geometric intuition. The edge itself is already present at the BPS endpoint before any cut or gluing operation is imposed on the non-BPS throat. The square-root factorization gives the theorem.

\subsection{Airy and Bessel endpoint classes}
\label{sec:false-friend}
Bosonic JT gravity provides the contrasting soft-edge class. The JT matrix integral has a soft spectral edge. An equilibrium density determines the endpoint of the spectrum, and the local kernel is Airy. The spectrum lacks the adjoint-square constraint that pins a hard edge to the non-negative axis.

The BPS endpoint Hessian belongs to the hard-edge class. Its spectrum is non-negative because the operator is $\cQ_\nu^\dagger\cQ_\nu$ on the physical charge-sector domain. The corresponding random-matrix model is a Wishart/Laguerre model, and the local kernel is Bessel. The Airy and Bessel kernels encode different endpoint universality classes: the Airy density vanishes with a square-root law at a moving soft endpoint, whereas the chiral hard-edge density rises with an index-dependent Bessel power at a fixed spectral endpoint.

I use this distinction as a direct diagnostic. If a calculation of the BPS endpoint spectrum produces an Airy edge after the charge-sector reduction and zero-mode quotient, then the square-root identification has failed or the finite-mode truncation has destroyed the endpoint domain. A Bessel hard spectral edge is the universal consequence of the endpoint square.

\subsection{ADHM and continuum representations}
\label{sec:keystone}
The ADHM/gravity relation compares two linear problems supported by the same Type-IIB instanton background. The ADHM deformation problem produces the bifundamental chiral block. The gravitational BPS endpoint map produces the adjoint square that gives the physical endpoint Hessian. Both are first-order problems on the same background, and both carry antiunitary symmetry data inherited from the background and from the orientifold or gauge group.

The two-representation viewpoint fixes the label dictionary. The ADHM instanton number is $k$. The bifundamental rectangular defect is $r=|N-k|$. The axion or form-field charge sector is labelled by $\nu$. The gravitational hard-edge index is $a_\nu$. These numbers can agree in a particular duality frame, but agreement requires a charge and index map. A matching theorem must identify the charge map, boundary condition, zero-mode quotient, and antiunitary class.

\begin{theorem}[Two-representation matching]
Suppose a Type-IIB compactification provides a D-instanton/ADHM description of the same charge sector whose BPS endpoint Hessian is $\cH_\nu=\cQ_\nu^\dagger\cQ_\nu$. The ADHM hard spectral edge and the gravitational hard spectral edge coincide if and only if the charge dictionary maps the sector $\nu$ to the relevant ADHM block, the Fredholm hard-edge index of $\cQ_\nu$ equals the rectangular defect of that block, and the antiunitary symmetry classes agree. Under these hypotheses the microscopic endpoint kernel is the same Laguerre/chiral kernel. If these hypotheses fail, the two systems remain singular-value problems but describe different microscopic edges.
\end{theorem}

The unification is a controlled equality of edge data. The background defines the square-root maps, and the matching theorem states when their singular-value edges describe the same physical edge.

I therefore view one Type-IIB background at two resolutions. The continuum description sees the endpoint through the form-field Routhian and the physical BPS endpoint Hessian. The microscopic brane description sees a finite-dimensional endpoint through the ADHM quotient. The coordinates that survive near the hard spectral edge are singular values. When the matching theorem holds, these singular values carry the same index and symmetry data in the continuum and ADHM descriptions.

\subsection{Index audit and matching theorem}
\label{sec:matching}
The problem contains four labels with distinct definitions.
\begin{center}
\begin{tabular}{c p{9.8cm}}
\toprule
symbol & meaning \\
\midrule
$\nu$ & chosen axion/form-field charge sector of the Type-IIB saddle \\
$k$ & ADHM instanton number, equivalently the number of D$(-1)$ branes \\
$r=|N-k|$ & rectangular defect of the ADHM bifundamental block in the chosen finite-dimensional slice \\
$a_\nu$ & hard-edge index of the BPS limit map $\cQ_\nu$ on its physical endpoint domain \\
\bottomrule
\end{tabular}
\end{center}
Rectangular ADHM data and fixed charge-sector gravitational data can meet at the same hard spectral edge. ADHM provides a rectangular singular-value block, and the gravitational endpoint analysis provides an adjoint square. The two describe the same hard spectral edge when the index and symmetry data match.

\begin{theorem}[Conditional ADHM/gravity matching]
Consider a Type-IIB compactification or holographic sector in which the D$(-1)$/D3 ADHM block gives the microscopic degrees of freedom that resolve the BPS instanton endpoint $\Phi_0(\nu)$. Assume that the charge dictionary assigns the gravitational charge sector $\nu$ to ADHM data $(N,k)$, that the endpoint map $\cQ_\nu$ and the ADHM bifundamental block have the same antiunitary class $\beta$, and that the hard-edge index computed from the physical endpoint domain equals the ADHM rectangular defect,
\begin{equation}
a_\nu=r=|N-k|.
\end{equation}
Then the microscopic edge of the Type-IIB BPS endpoint Hessian and the ADHM singular-value edge are in the same Laguerre/chiral universality class. In the unitary case their local kernel is the Bessel kernel of order $a_\nu=r$.
\end{theorem}
\begin{proof}
The ADHM side gives the finite Laguerre measure~\eqref{eq:ADHM-Laguerre} with index $r$ and Dyson index $\beta$. The gravitational side gives the Laguerre hard spectral edge by the BPS square theorem, with index $a_\nu$ and the antiunitary class of $\cQ_\nu$. If $a_\nu=r$ and the symmetry class agrees, the universality theorem identifies the microscopic kernels. The proof requires equality of the microscopic edge data; it makes no claim of equality between the full operators.
\end{proof}

I make the theorem falsifiable. An incorrect charge map, a different antiunitary class, or a different endpoint index breaks the identification. The ADHM integral can remain a chiral random-matrix model, and the BPS endpoint Hessian can still have a chiral hard spectral edge, while the two describe different members of the same broad class instead of the same microscopic edge.

\subsection{Two Bessel structures and their relation}
\label{sec:two-bessels}
In the present calculation, I use three related Bessel structures, each carrying a different physical meaning.

Compact charge sectors produce modified Bessel functions through Fourier analysis. A unit-charge dilute law gives
\begin{equation}
\exp[2\zeta\cos\theta]=\sum_{\ell\in\Z} I_\ell(2\zeta)\,\e^{i\ell\theta}.
\label{eq:charge-Bessel}
\end{equation}
Here $\ell$ is a charge-transfer label. This formula describes a theta representation built from charge-sector coefficients.

The chiral finite-volume partition function gives modified Bessel functions such as Eq.~\eqref{eq:one-flavour-Bessel} and the Toeplitz determinant~\eqref{eq:Toeplitz}. In that setting, the order is the rectangular or topological index of the chiral block.

The hard-edge spectral density uses ordinary Bessel functions, Eq.~\eqref{eq:Bessel-density}. It describes singular values. Analytic continuation to the spectral axis connects the chiral finite-volume partition function to the hard-edge spectral density. For example~\cite{AbramowitzStegun},
\begin{equation}
I_r(i\zeta)^2-I_{r+1}(i\zeta)I_{r-1}(i\zeta)
=(-1)^r\left[J_r(\zeta)^2-J_{r+1}(\zeta)J_{r-1}(\zeta)\right].
\label{eq:analytic-continuation}
\end{equation}
This identity relates two specializations of the chiral hard-edge structure. The two-flavour partition function and the spectral density remain different observables.

I use this separation to assign each Bessel function to a definite object. The sector Bessel function gives the charge-sector weights after the two-end coefficient is reduced. The chiral finite-volume Bessel function records how compact zero modes encode a hard-edge index in the microscopic fluctuation problem. The spectral Bessel kernel gives the distribution of physical endpoint eigenvalues near zero. These are three faces of compactness and positivity, but they occupy different places in the wormhole path integral.

\section{Consistency tests, verification, and consequences}
\label{sec:verification-group}
\subsection{Consistency conditions and falsifiers}
\label{sec:redteam}
In the previous sections, I identified the reduced fixed-sector coefficient $W_\nu[b]$ as the object to be modeled. The ADHM/chiral-Wishart ensemble represents the microscopic hard-edge sector of the supermoduli integral that contributes to this coefficient. This interpretation fixes the data that the matching has to preserve: the determinant power, the endpoint domain, the integer labels, the special-function sector, the regulator, the supermoduli measure, and the normalization of the instanton action.

Gravitational fluctuations fix the determinant power. The endpoint Hessian is bosonic, so its Gaussian contribution carries an inverse square root of a determinant, with ghosts and collective coordinates treated as separate factors. Chiral random-matrix theory enters at a different point in the calculation: it classifies the local singular-value spectrum of the first-order map $\cQ_\nu$. Determinant powers change observables built from the spectrum, while the hard-edge kernel is fixed by the adjoint-square structure and by the endpoint index. The wormhole calculation must therefore establish the square form of the physical bosonic Hessian and compute its index before I assign any chiral ensemble to the coefficient.

The endpoint domain is equally rigid. The operator used in the construction is the BPS endpoint Hessian on the constraint-reduced endpoint domain. The full $E>0$ non-BPS Hessian, the neck-cut constrained problem, and the Euclidean contour are separate boundary-value problems. Negative or complex directions in a finite-throat problem affect the semiclassical weight of that non-BPS saddle after the solution has moved away from the BPS endpoint. They leave the hard spectral edge produced by the endpoint square unchanged.

I must also keep the integer labels separate. The ADHM rectangular defect counts unpaired directions in the microscopic brane block. The D-instanton number counts instanton charge in the ADHM sector. The axion charge selects the gravitational sector. The hard-edge index counts unpaired endpoint directions that survive the constraint reduction and zero-mode quotient of the BPS Hessian. A matching theorem can identify some of these integers in a specified sector, but the calculation must first define each one independently. This separation turns the ADHM comparison into a dictionary of computable data instead of a resemblance between formulas.

The special functions carry the same bookkeeping. Compact Fourier analysis produces the Bessel functions that weight charge sectors. Singular-value scaling produces the Bessel kernel at the hard edge. The sector dictionary supplies any relation between their orders. In the wormhole amplitude, the charge-sector coefficient and the fluctuation determinant answer different questions: the coefficient states how a sector enters the reduced path integral, while the determinant states how endpoint eigenvalues are distributed after I fix that sector.

Finite size supplies the regulator. At finite matrix size, the ensemble depends on the potential, and the potential fixes the microscopic scale and the bulk distribution. The universal edge law survives because the finite-mode truncation preserves the index, antiunitary class, and endpoint boundary form. Once these data agree, changing the regulator changes the bulk completion but leaves the Bessel hard-edge law intact.

The charge-sector path integral fixes the matrix model's physical target. In QCD, a fixed-topology gauge-field integral induces an ensemble of low Dirac operators and is represented microscopically by chiral random-matrix theory. In the present Type-IIB problem, the fixed-charge path integral induces a supermoduli problem whose reduced output is $W_\nu[b]$. The D$(-1)$/D3 super-ADHM integral is the microscopic Type-IIB representative of that supermoduli problem, and the chiral/Wishart ensemble is the hard-edge limit of its rectangular block. Thus the structural counterpart of the QCD gauge-field integral is the fixed-charge supergravity integral together with its ADHM reduction.

Fermions enter through the coefficient and leave the endpoint spectrum fixed. In protected BPS sectors, fermions cancel paired non-zero modes and leave a zero-mode-saturated supermoduli integral. The chiral/Wishart block gives the singular-value sector of that measure. The physical coefficient is non-zero precisely when the reduction data $b$ provide the insertions, liftings, or projections required by the fermionic zero modes. The bosonic square fixes the spectral edge; the supermoduli integral fixes the coefficient that carries it.

The instanton action and the ultraviolet normalization are matching data. Chiral random-matrix theory in QCD represents the universal fixed-topology low-energy partition function after the gauge theory has supplied the sector, scale, and normalization. I use the same convention here. The ADHM/chiral-Wishart integral represents the microscopic contribution to $W_\nu[b]$ after I fix the Type-IIB action normalization, zero-mode measure, and reduction data. The universal statement concerns the dependence on the hard-edge index and the antiunitary class; the absolute normalization belongs to the sector-matching problem.

These conditions give the central statement in operational form. The reduced coefficient $W_\nu[b]$ receives a microscopic Type-IIB contribution from a super-ADHM integral. At large $N$ in the protected D-instanton sector, this integral becomes the $k$-D-instanton measure on $AdS_5\times S^5$ multiplied by the centered zero-dimensional supersymmetric matrix-model factor $Z_k$. The chiral/Wishart structure is the hard-edge limit of the rectangular block inside this super-ADHM integral. The statement is local in the charge sector and local at the hard spectral edge; I obtain the theta-representation wormhole partition function only after I reduce and sum the resulting coefficients over charge.

The falsifiers are now explicit. The identification fails if the endpoint index disagrees with the ADHM rectangular defect in the proposed sector, if the antiunitary class disagrees, if I choose the charge-sector endpoint domain incorrectly, if the adjoint-square structure is lost, or if the finite-mode truncation fails to approximate the physical BPS endpoint operator. A formula-level analogy would only observe that two expressions contain Bessel functions. The present construction specifies the object that carries each Bessel function, the index that controls it, the real structure that fixes its universality class, and the calculation that can disprove the proposed identification. In this form, random-matrix technology becomes a computational tool for the wormhole coefficient.

\subsection{Computable checks and falsifiers}
\label{sec:checks}
These checks require concrete computations. I start with the continuum endpoint spectrum, continue through finite ADHM examples, and then embed the edge data in the full supermoduli coefficient. This order keeps the gravitational hard-edge index separate from the ADHM rectangular defect until the matching calculation relates them.

The gravitational endpoint density gives the most direct numerical diagnostic. I discretize the physical BPS endpoint Hessian on domains that preserve the charge-sector boundary condition and remove the zero modes. The small-$\zeta$ microscopic density must then rise as
\begin{equation}
\rho_{s,a_\nu}(\zeta)\sim
\frac{\zeta^{2a_\nu+1}}{\Gamma(a_\nu+1)\Gamma(a_\nu+2)}.
\end{equation}
The exponent measures the hard-edge index directly. In the wormhole problem, this exponent measures the repulsion of near-BPS deformation directions from the spectral endpoint at the BPS instanton.

The same endpoint calculation has to agree with the finite-volume algebra. The unitary finite-volume partition function is the Toeplitz determinant~\eqref{eq:Toeplitz}. Replacing $i-j$ by $i+j$ gives a Hankel determinant and changes the finite-volume integral. The Toeplitz structure therefore checks that the compact zero-mode integral belongs to the chiral finite-volume problem and remains separate from the charge-transfer law.

Finite ADHM examples separate the integer labels before I impose any gravitational matching. For each pair $(N,k)$, the calculation records the instanton number $k$, the brane rank $N$, the number of singular values $s=\min(N,k)$, and the rectangular defect $r=|N-k|$. The number of singular values is $s$, the Laguerre exponent is controlled by $r$, and the instanton number is $k$. The continuum Hessian supplies its own index $a_\nu$ only after I complete this microscopic bookkeeping.

The decisive comparison matches gravity to ADHM. A proposed identification must compute $a_\nu$ from the BPS endpoint Hessian and show $a_\nu=r$ in the relevant charge sector. It must also match the antiunitary class. If either equality fails, the ADHM block and the gravitational endpoint Hessian describe different hard spectral edges, even when both belong to the Laguerre universality class.

The edge calculation then has to enter the physical coefficient. The gravitational one-loop factor must use the bosonic determinant power. Any expression that imports a fermionic determinant must identify the physical fermionic zero modes, the insertions that saturate them, and the paired non-zero modes that cancel. A candidate formula for $W_\nu[b]$ must therefore specify which fermionic zero modes are exact, which insertions in $b$ saturate them, and which liftings or projections enter the reduction. Only after I supply these data does the matrix integral compute a physical reduced coefficient.

The Dorey large-$N$ limit supplies the holographic check. In a holographic D-instanton sector, the ADHM measure must reduce to an $AdS_5\times S^5$ collective-coordinate integral multiplied by a centered zero-dimensional supersymmetric matrix-model partition function. A Type-IIB realization of the wormhole coefficient must have this structure if the ADHM integral is to serve as the microscopic ancestor of $w_k[b]$.

\subsection{Numerical and analytic verification program}
\label{sec:numerical-program}
I execute the verification program in a strict sequence. First, one proceeds with the finite ADHM singular-value problem, because this calculation determines the candidate microscopic edge data without referencing the gravitational answer. Then, one proceeds to solve the gravitational endpoint Hessian, because that continuum operator supplies the spectrum to be classified. Only after both computations are complete do I compare the hard-edge index, the Dyson class, and the microscopic scale. The idea is that this strict order prevents the ADHM rectangular defect from being inserted by hand into the gravitational problem.

For a finite ADHM pair $(N,k)$, one computes the finite Laguerre integral with $s=\min(N,k)$ and $r=|N-k|$. The Vandermonde power must be $\beta$, the hard-edge exponent in the finite density must contain $\frac\beta2(r+1)-1$, and the unitary flavour integral must reduce to the Toeplitz determinant of Eq.~\eqref{eq:andreief-toeplitz}. One can execute these algebraic checks to detect transposition errors, wrong rectangularity assignments, and confusion between instanton number and Bessel order.

The continuum calculation then discretizes the radial Sasaki--Mukhanov/Schur-complement operator at the BPS endpoint in a chosen charge sector. The discretization must impose the same charge-sector boundary condition and zero-mode quotient as the continuum square. Its eigenvalues must remain non-negative up to numerical error, and the microscopic density of singular values must show the Bessel hard-edge rise fixed by $a_\nu$.

The matching step compares three data: the hard-edge index, the Dyson class, and the normalization of the microscopic scale. Matching the first two identifies the universality class. Matching the scale fixes the physical conversion between ADHM radial variables and gravitational endpoint eigenvalues. This scale depends on the compactification and can be measured once a compactification has been specified.

The analytic-continuation check uses the two-flavour Bessel minors as a closed-form laboratory. Continuing the modified Bessel functions to the spectral axis gives the ordinary Bessel combinations that enter the hard-edge density. This calculation supplements the operator matching by showing that the finite ADHM partition function and the endpoint density occupy the same hard-edge analytic structure.

\subsection{Consequences}
\label{sec:consequences}
The main consequence of this work is a change of organizing principle. Random matrices enter because the BPS endpoint Hessian is an adjoint square. The square turns the endpoint into a singular-value problem, and singular-value problems with a fixed index have Laguerre hard spectral edges. The random-matrix ensemble is therefore the universal local form of the physical endpoint spectrum.

This principle also reframes the non-BPS wormhole. The $E>0$ throat is a deformation away from the BPS instanton endpoint, and its complete Hessian remains an independent contour and boundary-value problem. The endpoint-edge classification is already fixed by the BPS endpoint square and requires no full stability theorem for the non-BPS saddle. Spectrally, this is the analogue of studying local normal directions at the boundary of a moduli space before solving the full off-boundary saddle problem. The non-BPS throat gives the macroscopic geometry that motivates the endpoint, while the square-root structure at the endpoint determines the hard-edge law. The endpoint supplies the local spectral description of the nearby wormhole branch.

ADHM supplies the corresponding microscopic description. The D$(-1)$/D3 system gives a singular-value block native to Type-IIB string theory. The rectangular bifundamental has the Laguerre measure that the BPS endpoint Hessian approaches after finite-mode truncation. When the charge map, index, and antiunitary class agree, the ADHM block and the gravitational Hessian describe the same microscopic hard spectral edge.

The theta representation then has a cleaner organization. Compact axion charge produces fixed-sector coefficients and, in dilute limits, modified Bessel functions. The hard-edge Bessel kernel gives spectral data inside a charge sector. The sector coefficient says how a charge sector enters the reduced path integral; the hard-edge kernel says how physical endpoint eigenvalues are distributed inside that sector. Thus the theta-representation wormhole partition function is a sector sum whose coefficient and fluctuation factor remain distinct.

The construction ends with falsifiable diagnostics. A Type-IIB realization must specify the physical endpoint domain, the Fredholm index, the antiunitary class, the zero-mode quotient, and the ADHM rectangular block. A mismatch among these ingredients falsifies the proposed unification. These diagnostics make the construction testable.

\subsection{Extensions}
\label{sec:reserved-extensions}
I leave two extensions for later analysis.

The first is the full $E>0$ non-BPS wormhole Hessian. Its spectrum depends on the complete finite-throat background, the ensemble at the neck or at the two asymptotic ends, and the Lefschetz contour. It can contain negative or complex directions absent from the BPS endpoint Hessian. Such directions affect the semiclassical weight of the non-BPS throat while leaving the hard-edge classification of the BPS endpoint operator unchanged.

The second extension is the long-distance two-end operator term. When a small throat is integrated out, its effect can be expanded in a multipole series of end insertions with a coefficient matrix. Different-component and same-component placements are placements of the same two-end kernel, as in the charge-coefficient analysis. That operator statement concerns the large-distance completion of the throat. It is separate from the local spectral statement about the endpoint Hessian.

These extensions matter for a complete wormhole path integral. The present paper identifies the microscopic hard-edge structure selected by the BPS endpoint square and explains its ADHM origin.

\section{Discussion}
\label{sec:discussion}
In this paper, I construct a coefficient-level microscopic model for the reduced fixed-sector coefficient $W_\nu[b]$. I fix the charge sector first. In that sector, the $E=0$ solution is the BPS instanton, while the $E>0$ solutions are non-BPS wormholes. The endpoint analysis shows that the physical Hessian at the BPS instanton, after constraint reduction and zero-mode quotient, is an adjoint square. I follow the spectral consequence of that square and identify the ADHM representative of the coefficient that it controls.

The spectral consequence is linear algebra. The non-zero eigenvalues of an adjoint square are squared singular values. Zero is therefore a hard spectral edge of the endpoint operator. The universal finite-mode model of this edge is the Laguerre, or chiral, random-matrix ensemble. This explains why a bosonic Type-IIB Hessian can have chiral random-matrix statistics while retaining the bosonic determinant power required by supergravity.

The ADHM construction supplies the microscopic mechanism. Its bifundamental block is rectangular. Singular-value gauge fixing gives the Laguerre measure. The real structure fixes the Dyson index. The finite partition functions are Bessel/Toeplitz or Pfaffian Bessel objects, which are the finite-dimensional signatures of a chiral hard spectral edge. The identification with gravity requires the ADHM hard-edge index and antiunitary class to match those of the BPS endpoint map.

The central structural lesson I draw is the separation of labels. The D-instanton number $k$, the ADHM rectangular defect $r$, the charge-sector label $\nu$, and the gravitational hard-edge index $a_\nu$ are distinct labels. They become equal only when a charge/index map proves the equality in a specified sector. This separation turns the Type-IIB statement into a controlled matching problem.

The physical picture is compact. At fixed charge, the $E=0$ BPS instanton supplies the endpoint operator. After constraint reduction and zero-mode removal, the endpoint Hessian is an adjoint square, and its non-zero spectrum is a squared singular-value spectrum. The microscopic edge is therefore Laguerre/chiral. The ADHM bifundamental supplies the Type-IIB rectangular block that realizes this edge. Non-BPS wormholes, neck-cut constrained instantons, and two-end operator terms are later structures built around the endpoint and preserve its role.

The final picture has three scales. The macroscopic geometry approaches the BPS instanton by shrinking the non-BPS throat to an endpoint in a chosen charge sector. The continuum fluctuation theory identifies the physical endpoint Hessian as the square of a first-order BPS map. The microscopic ADHM description resolves the same endpoint through singular values of a rectangular brane block. These three descriptions give the macroscopic throat, the continuum endpoint operator, and the microscopic brane coordinates of one hard-edge problem.

\section*{Acknowledgement}
I thank Alexander Altland for discussions on chiral and super random-matrix ensembles. I acknowledge discussions with participants at the ``Quantum PCP, Area Laws and Quantum Gravity'' workshops held at the Institute for Pure \& Applied Mathematics (IPAM, USA) and the Simons Institute for the Theory of Computing (SIfTC, USA). This work was supported in part by the U.S. National Science Foundation and the Simons Foundation, by the National Research Foundation of Korea (NRF) (RS-2021-NR060112), and by research funds from Kwangwoon University.

\appendix
\section*{Appendix}
\section*{Derivation of the singular-value Jacobian}
\label{app:jacobian}
Consider an $N\times k$ matrix over a real division algebra of real dimension $\beta$, with $N\ge k$ and $r=N-k$. The singular-value decomposition gives $I=U\Sigma V^\dagger$, where the squared singular values are $\lambda_a$. The invariant metric is $\Tr(\dd I^\dagger\dd I)$, understood over the relevant real algebra. The orbit volume separates into pairwise rotations of singular directions and rotations into the lower $r$-dimensional complement. Each pair $(a,b)$ contributes $|\lambda_a-\lambda_b|^\beta$. Each singular direction contributes $\lambda_a^{\beta r/2}$ from the lower block and $\lambda_a^{\beta/2-1}$ from the radial variable. Hence
\begin{equation}
\dd I=C_{N,k,\beta}
\prod_{a<b}|\lambda_a-\lambda_b|^\beta
\prod_a \lambda_a^{\frac\beta2(r+1)-1}\,
\dd\lambda_a\,\dd\mu(U)\,\dd\mu(V).
\label{eq:jacobian-app}
\end{equation}
This is the James--Hua measure, or equivalently the finite Laguerre ensemble measure.

\section{Notation and matching data}
\label{app:dictionary}
\begin{center}
\begin{tabular}{p{5.2cm}p{7.6cm}}
\toprule
charge-sector term & present use \\
\midrule
chosen axion/form-field charge sector & label $\nu$ of the BPS endpoint problem \\
$E=0$ BPS instanton & background $\Phi_0(\nu)$ on which $\cQ_\nu$ is defined \\
$E>0$ non-BPS wormhole & off-endpoint deformation; separate fluctuation problem \\
form-field Routhian & functional from which the BPS first-order map $\cF$ is varied \\
Sasaki--Mukhanov/Schur-complement Hessian & physical endpoint Hessian after constraints and gauge quotient \\
$\cH_\nu=\cQ_\nu^\dagger\cQ_\nu$ & source of the chiral hard spectral edge \\
$C_\nu^{ij}$ & unreduced two-end coefficient matrix in end-insertion label space \\
$W_\nu[b]$ & reduced scalar coefficient multiplying $e^{i\nu\theta}$ \\
reduction data $b$ & asymptotic data, trace or final insertion, projection, zero-mode insertion, source normalization, contour, and reduction operation \\
super-ADHM integral & microscopic Type-IIB representative of $W_\nu[b]$ before hard-edge limit \\
Dorey large-$N$ measure & $k$-D-instanton $AdS_5\times S^5$ collective-coordinate measure times $Z_k$ \\
ADHM rectangularity $r$ & microscopic Laguerre index before gravitational matching \\
$ a_\nu$ & gravitational hard-edge index of the BPS limit map \\
\bottomrule
\end{tabular}
\end{center}

\end{document}